\documentclass[letterpaper,onecolumn,11pt,unpublished]{quantumarticle}
\pdfoutput=1
\usepackage[T1]{fontenc}
\usepackage{amsmath,amssymb,amsthm,mathtools,booktabs}
\usepackage{microtype}
\usepackage[numbers,sort&compress]{natbib}
\usepackage[colorlinks=true,linkcolor=blue,citecolor=blue,urlcolor=blue]{hyperref}
\makeatletter
\long\def\@makecaption#1#2{%
  \vskip\abovecaptionskip
  \sbox\@tempboxa{#1: #2}%
  \ifdim \wd\@tempboxa >\hsize
    #1: #2\par
  \else
    \global \@minipagefalse
    \hb@xt@\hsize{\hfil\box\@tempboxa\hfil}%
  \fi
  \vskip\belowcaptionskip}
\makeatother

\newtheorem{theorem}{Theorem}
\newtheorem{proposition}{Proposition}
\newtheorem{lemma}{Lemma}
\newtheorem{corollary}{Corollary}
\theoremstyle{definition}
\newtheorem{definition}{Definition}
\theoremstyle{remark}
\newtheorem{remark}{Remark}

\newcommand{\R}{\mathbb R}
\newcommand{\E}{\mathbb E}
\newcommand{\ket}[1]{\lvert #1\rangle}
\newcommand{\ind}[1]{\mathbf 1_{\{#1\}}}
\newcommand{\sgn}{\operatorname{sgn}}
\newcommand{\tr}{\operatorname{tr}}
\newcommand{\diag}{\operatorname{diag}}

\title{Quantum ordinal optimization under a fixed comparison budget}
\author{Hongyu Wu}
\affiliation{Mike Wiegers Department of Electrical and Computer Engineering, Kansas State University,
Manhattan, KS 66506, USA}
\date{}

\begin{document}
\maketitle

\begin{abstract}
Ordinal optimization selects designs that are good enough rather than
optimal when evaluation is expensive and the budget is small. The alignment
probability measures the chance that the selected set contains at least a
prescribed number of truly good designs. We formulate quantum ordinal
optimization as a class of problems that maximize this probability under a
fixed oracle budget. We study exact, coherent pairwise comparisons with one selected design as a first
instance, with the budget counted in comparison queries. For every classical comparison budget smaller than the number of designs,
we prove that a knockout among uniformly sampled designs attains the optimum.
The sample is drawn without replacement and contains one more design than
the number of comparisons. For a uniformly random order, the selected design's rank has a polynomial
probability law. At the same comparison budget, the maximum quantum degree is at most twice the classical
maximum, a factor attained whenever there are more than twice as many designs as
comparisons, so a quantum algorithm can realize rank laws that no classical algorithm
with the same budget can. We determine the quantum optimum with one
comparison for every number of designs and every good-enough set. For
selecting the best of $n$ designs, the optimal success probability with one
quantum comparison is approximately $3.30/n$ for large $n$, compared with
$2/n$ for one classical comparison. When the best design is sought,
one quantum comparison gives a higher alignment probability than two
classical comparisons whenever there are more than nine designs.
\end{abstract}

\section{Introduction}

Many design and operational decisions require a choice from a large finite set of
designs whose performance is estimated through expensive simulation.
Ordinal optimization (OO) addresses such problems by accepting good-enough
designs and comparing their order rather than their
values~\cite{HoSreenivasVakili1992,HoZhaoJia2007}. OO measures selection
success by the alignment probability, the probability that the selected
set contains at least a prescribed number of truly good-enough designs.
This required number is called the alignment level~\cite{LauHo1997}.
In constrained OO, a selected design needs to be feasible. Constrained OO
has been applied to stochastic security-constrained unit commitment with
high wind penetration, where a feasibility model guides a
blind pick (sampling designs uniformly at random without comparing their performance) among candidate schedules~\cite{WuShahidehpour2014}.
Recent studies continue to apply OO to stochastic job-shop
scheduling~\cite{GhasemiAshooriHeavey2021}, emergency-department resource
allocation~\cite{HorngLin2024}, and task scheduling in wireless edge
computing~\cite{WangSun2025}.

Quantum minimum finding shows that coherent access can reduce the number
of queries needed for selection~\cite{DurrHoyer1996}. A coherent comparison
acts on a superposition of design pairs, and the resulting amplitudes can
interfere before measurement. An algorithm can use this interference to
favor better ranks in its output distribution. For OO, the question is
whether this capability improves alignment probability at a fixed
comparison budget and how any improvement depends on the size of the
good-enough set. However, accepting more designs as good enough also makes
classical selection easier. Assessing a quantum benefit therefore requires
comparison with the classical optimum for the same target and comparison budget.

Asymptotic query bounds for quantum
selection~\cite{DurrHoyer1996,HuHuZhou2026} do not determine the exact
alignment probability at a small, fixed comparison budget. Fixed-budget
quantum simulation optimization has also been studied through convergence
rates for the expected optimality gap under quantum simulation
access~\cite{HuHu2026}. That criterion measures loss in performance value,
whereas alignment measures whether the selected rank meets the target.
To the best of our knowledge, prior work has not determined the exact
alignment optimum over all admissible classical or quantum strategies in
the comparison model. This optimum is needed to distinguish a benefit of
coherent access from a weakness of a particular classical selection rule.

To address this gap, we formulate quantum ordinal optimization (QOO)
as a class of problems that maximize the worst-case alignment probability
under a specified oracle model and query budget. The oracle, selected-set
size and alignment level are part of the problem specification. As a first
instance, we study exact coherent pairwise comparisons with one selected
design and alignment level one. Exact comparisons separate selection from
estimation error, while a single output reduces alignment to meeting the
rank target. For each
budget and good-enough set, the objective is the largest alignment
probability that can be guaranteed for every strict order of the designs.
The budget counts only comparison queries. Estimation error and
implementation cost require further models, as discussed in
Section~\ref{sec:program}.

In this comparison problem, we establish the classical optimum for every
budget smaller than the number of designs and determine the one-comparison
quantum optimum. One comparison is the smallest budget at which selection
can depend on the unknown order, yet it already permits a quantum gain.
We report each value as an equivalent blind-pick size, the number of designs that a
blind pick would need to reach the same alignment probability, which need not be an
integer.
Our main results are as follows.
\begin{enumerate}
\item For every good-enough set and every budget smaller than the number of
designs, we prove that the optimal alignment probability with exact classical
comparisons equals that of a blind pick whose size exceeds the budget by one.
A knockout among the blindly picked designs attains this probability.
Although the comparisons can involve twice as many designs as their number,
no adaptive procedure that returns one design does better than this blind
pick when comparisons are exact.
\item Averaged over a random order, the law of the rank of the design
selected by a quantum algorithm is a polynomial in the rank. Its degree is at
most twice the
comparison budget or one less than the number of designs, whichever is
smaller, and this bound is attained. By contrast, the degree for classical algorithms is at
most the comparison budget. At one comparison the larger quantum degree
is necessary for any gain.
\item We determine the exact quantum optimum with one comparison for every
number of designs and every good-enough set, with a matching upper
certificate and an attaining algorithm. The value is affine in the smallest
positive root of an explicit cubic, with closed forms for the best design
and for the top half, and it gives the limiting success probability for a fixed good-enough fraction.
\item For selecting the best design, the equivalent blind-pick size of one quantum comparison increases with the number of designs toward
$(3+\sqrt{13})/2\approx3.30$, compared with $2$ for one classical
comparison, and it exceeds that of two classical comparisons exactly when
there are more than nine designs. For the top $5$ percent of $1000$ designs
it is about $3.27$. In the cases computed it decreases as the good-enough set
grows, and for the top half of $1000$ designs it is below that of two
classical comparisons.
\end{enumerate}

The one-comparison result also raises the question of how the quantum
benefit changes as more comparisons are allowed.
Known search bounds and the reduction in Section~\ref{sec:program-budgets}
show that for selecting the best design the quantum equivalent blind-pick
size grows at most quadratically in the comparison
budget~\cite{Zalka1999,BennettEtAl1997}, and minimum finding attains this
order when the number of designs is at least of the order of the square of
the budget~\cite{DurrHoyer1996}. The classical size is the budget plus one, so the ratio of quantum to classical sizes
grows at most in proportion to the budget and reaches that order when the designs are
numerous enough. This paper determines the ratio exactly for one comparison.

Section~\ref{sec:prior} reviews related work, and Section~\ref{sec:model} defines the
model and measures. Sections~\ref{sec:classical}--\ref{sec:one} establish the classical
optimum, the quantum degree bound and the exact one-comparison value.
Section~\ref{sec:oo} compares the values in OO terms, Section~\ref{sec:program}
discusses larger budgets, evaluations and computational cost, and
Section~\ref{sec:conclusion} summarizes the findings.

\section{Relation to prior work}\label{sec:prior}

The relevant literature treats selection from performance estimates,
selection through comparisons and exact optimization of quantum query
algorithms. We distinguish these settings because their oracles and success
criteria determine which guarantees can be compared with OO alignment.

\paragraph{Classical ordinal optimization.}
OO uses the fact that an estimated order can converge faster than estimated
values. For several common performance measures, the order converges
exponentially fast in the simulation length~\cite{Dai1996}. Its selection
rules include blind pick, which selects designs at random; the horse race,
which selects the designs that look best under a crude
model~\cite{LauHo1997}; and budget allocation rules such as optimal computing
budget allocation (OCBA)~\cite{ChenLinYucesanChick2000}. These rules are
catalogued and compared in~\cite{HoZhaoJia2007}.

Classical OO compares selection rules mainly by simulation. Lau and Ho
listed round-robin and pairwise comparison rules as yet to be
studied~\cite{LauHo1997}. These rules were later defined for noisy
observations and compared numerically, and no universally best selection
rule is known~\cite{HoZhaoJia2007}. Apart from budget allocation rules such as OCBA, existing OO optimality results apply to
the final selection given the observations. Under independent, identically
distributed noise and equal allocation, choosing the observed top designs
is optimal~\cite{HoZhaoJia2007}, and with design-dependent noise the optimal
rule ranks designs by their probability of being good
enough~\cite{YangLee2002}. These results do not optimize which comparisons
to make.

\paragraph{Quantum selection from estimated values.}
For quantum best-arm identification with reward gaps $\Delta_i$, order
$\sqrt{\sum_i\Delta_i^{-2}}$ queries suffice and are necessary, up to
polylogarithmic factors~\cite{WangYouLiChilds2021}. Weaker oracle models
change the rate~\cite{WangEtAl2025}. Whether a speedup exists depends on
access to the simulator's randomness~\cite{BuchholzKublerScholkopf2025}.
With a weak quantum oracle, a fixed-confidence algorithm identifies a given
number of best arms as a set~\cite{WangEtAl2025QCE}.
Quantum Grover adaptive search for discrete simulation optimization returns an
$\epsilon$-optimal design with probability at least $1-\delta$ using
$\tilde O(\sqrt n)$ queries for constant $\epsilon$ and
$\delta$~\cite{HuHuZhou2026}. That is a fixed-confidence guarantee.
Hu and Hu study fixed-budget simulation optimization through the minimax
expected optimality gap and establish convergence rates under quantum
simulation access~\cite{HuHu2026}. Both formulations use performance
values estimated from a simulator. Here success depends only on the
selected rank, and we determine exact finite-budget values in a
comparison-oracle model.

\paragraph{Selection with comparisons.}
Minimum finding uses order $\sqrt n$ comparisons against a moving
threshold~\cite{DurrHoyer1996}, and finding the $r$ smallest values needs
order $\sqrt{rn}$ for $r\le n/2$~\cite{DurrHeiligmanHoyerMhalla2006}.
Bounds for approximating an order statistic give order $\sqrt{n/r}$, up to
logarithmic factors, for an output among the top $r$~\cite{NayakWu1999}.
Grover adaptive search has been analyzed for reaching one of the best $r$
points~\cite{BaritompaBulgerWood2005}. That work bounds the expected cost
and computes the success probability numerically as a function of the
number of function evaluations for two given algorithms. Comparison oracles
on permutations are the standard setting of quantum lower bounds for
sorting and element distinctness~\cite{HoyerNeerbekShi2002}. Wang and Ying use a bounded-error version of this comparison oracle~\cite{WangYing2024}.

\paragraph{Exact results with few queries.}
Grover's algorithm is exactly optimal for search with one marked
item~\cite{Zalka1999}. Zero-error few-query algorithms for ordered search were
found numerically~\cite{FarhiEtAl1999,ChildsLandahlParrilo2007}, and
translation-invariant algorithms are optimal for that
problem~\cite{CarolanEtAl2025}. For small Boolean functions, optimal success probabilities have been computed
numerically by semidefinite programming in a study of zero-error quantum query complexity~\cite{MontanaroJozsaMitchison2015}, and optimal values are
known for oracle problems with a group
structure~\cite{CopelandPommersheim2021} and for one-query AND and
EQUALITY~\cite{AmbainisIraids2017}. Several of these results use the
semidefinite characterization of quantum query
algorithms~\cite{BarnumSaksSzegedy2003}.

The comparison results provide query bounds for selection, while the
few-query results provide methods for optimizing success at a specified
budget. Below, we formulate rank alignment under comparison access,
establish the exact classical benchmark and apply these methods to the
one-comparison quantum problem. Selection from estimated values requires
a different evaluation model, as discussed in
Section~\ref{sec:program-models}.

\section{Model and performance measures}\label{sec:model}

Determining the best alignment probability at a fixed budget requires a
QOO problem to specify what information queries provide and what counts as
a successful selection. We first define the problem class with these choices
left to each problem. We then specialize it to exact comparisons with one
selected design, for which the following sections establish the results.

\subsection{Quantum ordinal optimization}\label{sec:qoo-general}
\begin{definition}[Quantum ordinal optimization]\label{def:qoo}
A quantum ordinal optimization problem specifies a class of unknown
configurations, each of which determines a strict order of $n$ designs and
may contain further data such as performance values or gaps. It specifies
the good-enough set of the $r$ best designs, an oracle whose responses depend
on the configuration, a budget of at most $q$ calls on every branch of an algorithm's measurements and private randomness, and
the size of the selected set and its alignment level. The oracle model
specifies its quantum action and the response to a classical query. An
algorithm succeeds when its output set contains at least the prescribed
alignment level of good-enough designs.

An admissible quantum algorithm uses input-independent state preparation,
unitary operations, measurements, private randomness and finite workspace.
It has no postselection and no access to the configuration except through
the specified oracle. The classical counterpart uses private randomness
and adaptive queries to the same oracle on one classical input at a time.
Both obey the prescribed query budget and selected-set size. The quantum
value is the supremum, over admissible quantum algorithms, of the infimum
of their success probabilities over the configuration class. The classical
value is defined in the same way over classical algorithms.
\end{definition}

The oracle can provide comparisons, evaluations or membership information,
and the selected set can contain one or several designs. Each choice
specifies a QOO problem with its own classical benchmark. For a fixed oracle and selected-set size, the quantum value of a QOO problem is the
optimal success probability of a query problem at a fixed number of queries, the
quantity studied for other problems in the few-query work of Section~\ref{sec:prior}.
What OO adds is the good-enough target with its alignment level, the size of the
selected set, and the comparison with the optimal classical procedure and with blind
pick. Section~\ref{sec:program-models} discusses how changes in evaluation and selection
affect the formulation.

\subsection{The exact-comparison instance}\label{sec:comparison-instance}
We study the instance of Definition~\ref{def:qoo} whose configurations are
all strict orders, with the exact comparison oracle, selected-set size one
and alignment level one. Classical and quantum procedures use the same
comparison rule and good-enough target. They differ in how they query that
rule, and this instance requires neither performance values nor numerical gaps
between designs.

Let $[n]=\{1,\ldots,n\}$ label $n$ designs. Their unknown performance order
is a strict order, encoded by the rank vector $x=(x_1,\ldots,x_n)$, a
permutation of $[n]$ in which rank $1$ is best. An algorithm receives labels and
oracle answers, never ranks. For $0\le r\le n$ the good-enough set is the
true top $r$, $G_r(x)=\{i:x_i\le r\}$. The algorithm outputs one label $I$,
and it succeeds when $I\in G_r(x)$. In OO terms the selected set is $\{I\}$,
of size one, the alignment level is one, and the alignment probability is
$\Pr[x_I\le r]$~\cite{LauHo1997}.
For $r=1$, this is the probability of correct selection (PCS) in ranking
and selection~\cite{HongFanLuo2021}.
For $r>1$, it is a rank-based form of the probability of good selection
(PGS), where goodness requires rank at most $r$, whereas Eckman and
Henderson define good selection by a mean performance gap from the best
that is smaller than a prescribed tolerance~\cite{EckmanHenderson2018}.

\paragraph{Comparison oracle.}
A quantum comparison is the oracle
\begin{equation}\label{eq:oracle}
O_x\ket{i,j,c,w}=\ket{i,j,c\oplus\ind{x_i<x_j},w},\qquad i,j\in[n],
\end{equation}
where $c$ is a response bit, $w$ is workspace, and $\ind{x_i<x_i}=0$. This
is the standard comparison oracle on permutations~\cite{HoyerNeerbekShi2002,WangYing2024}.
A classical comparison queries one pair and receives the bit $\ind{x_i<x_j}$.
An admissible quantum algorithm in this instance obeys
Definition~\ref{def:qoo}, makes at most $q$ calls to~\eqref{eq:oracle} on
every branch and outputs one label. Its only access to $x$ is through this
oracle. A classical algorithm is an adaptive comparison tree with private
randomness and at most $q$ comparisons on every branch. Only oracle calls
are charged.

\begin{remark}
A self-address $i=j$ returns no information about $x$. It is allowed for
convenience, and no value in this paper depends on it, because every call
already acts as the identity on the $\ket+$ eigenspace of the response bit.
\end{remark}

\begin{definition}
For $q\ge0$, $n\ge1$ and $0\le r\le n$, the quantum value is
\begin{equation}\label{eq:value}
Q_q(n,r)=\sup_A\ \min_x\ \Pr_x^A[x_I\le r],
\end{equation}
where $A$ ranges over admissible algorithms with budget $q$ and $x$ over
strict orders. The classical value $C_q(n,r)$ is defined in the same way
over classical algorithms.
\end{definition}

\subsection{Worst case and average case}
The worst-case objective can be analyzed through a uniform random order
because an algorithm may randomly relabel designs before making
comparisons. The same relabeling is used throughout the run, so the
comparison answers remain consistent.
\begin{proposition}\label{prop:twirl}
For every admissible algorithm $A$ there is an admissible algorithm $A'$
with the same budget whose success on every order equals the success of $A$
averaged over a uniformly random order, and whose whole rank law has the
same property. Consequently $Q_q(n,r)$ and $C_q(n,r)$ equal the suprema of
the uniform-average success.
\end{proposition}
\begin{proof}
Let $A'$ draw a uniform permutation $\sigma$ of $[n]$ from private
randomness, and let
\[
P_\sigma\ket{i,j,c,w}=\ket{\sigma(i),\sigma(j),c,w}.
\]
Then $P_\sigma^{-1}O_xP_\sigma=O_{x\circ\sigma}$, where
$(x\circ\sigma)_i=x_{\sigma(i)}$. Let $A'$ run $A$ with every call $O_x$
replaced by $P_\sigma^{-1}O_xP_\sigma$ and output $\sigma(I)$. On a fixed
order $x$, $A'$ therefore behaves as $A$ on the order $x\circ\sigma$, which is
uniformly distributed, and the rank of $\sigma(I)$ under $x$ equals the rank
of $I$ under $x\circ\sigma$. The same $\sigma$ is used for every call and for
the output, so every rank event is uniformized at once. The worst-case
success of any algorithm is at most its average success, which gives the
second statement.
\end{proof}
We average the algorithm over relabelings in Proposition~\ref{prop:twirl}, as in
the symmetrization of semidefinite programs for ordered
search~\cite{CarolanEtAl2025}, whereas Childs and Lee average adversary matrices with the automorphism
principle~\cite{ChildsLee2008}. A fresh relabeling between calls
would not preserve the oracle answers, and it is not used.

\subsection{Blind pick and the equivalent blind-pick size}
A blind pick selects $s$ designs uniformly at random without comparing them.
Its selected set contains a design of the good-enough set with probability
\begin{equation}\label{eq:blind}
B(n,r,s)=1-\frac{\binom{n-r}{s}}{\binom ns},
\end{equation}
the blind-pick alignment probability at alignment level
one~\cite{HoZhaoJia2007,LauHo1997}. When a second stage identifies the best
of the $s$ designs exactly, as in the two-stage use of OO~\cite{HoZhaoJia2007},
that design is good enough with the same probability. For $r=1$,
$B(n,1,s)=s/n$.

For $1\le r\le n$ and with
$\binom ab=\Gamma(a+1)/\bigl(\Gamma(b+1)\Gamma(a-b+1)\bigr)$, the ratio
in~\eqref{eq:blind} extends to real $s$ in $[0,n-r+1]$, with
$B(n,r,n-r+1)=1$. The logarithmic derivative of the ratio in $s$ is
$\psi(n-r-s+1)-\psi(n-s+1)<0$, where $\psi$ is the digamma function, so
$B(n,r,s)$ increases strictly in $s$ from $0$ to $1$.
\begin{definition}
For $1\le r\le n$, the equivalent blind-pick size of a value $P\in[0,1]$ is the unique
$s_{\rm eq}(n,r)\in[0,n-r+1]$ with $B(n,r,s_{\rm eq})=P$. We write
$s_{\rm eq}^{Q}(n,r)$ for $P=Q_q(n,r)$ and $s_{\rm eq}^{C}(n,r)$ for
$P=C_q(n,r)$. For $r=1$, $s_{\rm eq}^{Q}(n,1)=n\,Q_q(n,1)$.
We leave the dependence on $P$ and, for the classical and quantum values,
on $q$ implicit in this notation.
\end{definition}
Using equivalent blind-pick sizes follows the OO practice of comparing
selection rules by the selected-set size needed for a given alignment
probability~\cite{HoZhaoJia2007}.

\paragraph{Scope.} Comparisons are exact, meaning noiseless, and the oracle is coherent, so the
algorithm may query superpositions of pairs. The budget counts oracle calls.
Exact comparisons model problems in which each design is evaluated
deterministically but at high cost, such as a unit commitment evaluated against a fixed set of scenarios. Whether such an
evaluation can be queried coherently, and at what cost, is the question of
Section~\ref{sec:program-cost}. When performance is observed only through
noisy simulation, an exact comparison cannot be implemented
directly~\cite{HuHuZhou2026}. Section~\ref{sec:program-models} discusses
extensions to that setting.

The choice of cost unit also affects the interpretation. In a model that
charges evaluated designs, a blind pick followed by keeping the best is the
natural rule. With comparisons as the cost unit,
Theorem~\ref{thm:classical} below shows that $q$ comparisons attain exactly
the alignment probability of $q+1$ evaluated designs, although the comparisons can involve up to $2q$ designs. The results of this paper count
comparison queries and make no claim about running time. If a coherent comparison is
built from an oracle that evaluates one design, computing and uncomputing both
evaluations uses four evaluation calls, and four evaluated designs already contain the
best design with probability $4/n$. The one-comparison gain of Sections~\ref{sec:one}
and~\ref{sec:oo} is therefore a gain per comparison and not per evaluation. The order of
growth discussed in Section~\ref{sec:program-budgets} does not change under this unit,
because $q$ comparisons then cost $4q$ evaluations while the quantum equivalent
blind-pick size grows quadratically in $q$. The gate cost of an oracle call
requires a further cost model, as discussed in Section~\ref{sec:program-cost}.

\section{The classical optimum}\label{sec:classical}

The blind-pick formula gives the success probability of a knockout on a
random subset with one more design than the comparison budget. We first
show that no adaptive comparison procedure can
improve this probability with the same budget. This establishes the
classical benchmark needed to assess the benefit of coherent access.

\begin{theorem}\label{thm:classical}
For integers $q\ge0$, $n\ge q+1$ and $0\le r\le n$,
\begin{equation}\label{eq:classical}
C_q(n,r)=1-\frac{\binom{n-r}{q+1}}{\binom n{q+1}}=B(n,r,q+1).
\end{equation}
The following algorithm attains this value on every order and for every $r$
at once. Choose $q+1$ distinct labels uniformly at random, compare the
current winner with the next label $q$ times, and output the final winner.
Its output rank law is
\begin{equation}\label{eq:classical-law}
\Pr[x_I=t]=\frac{\binom{n-t}{q}}{\binom n{q+1}},\qquad t\in[n].
\end{equation}
\end{theorem}

To prove optimality, we examine the orders consistent with the comparison
answers along a branch of an arbitrary comparison tree. The following
lemma describes the remaining rank uncertainty in the connected component
containing the output label.

\begin{lemma}\label{lem:posterior}
Fix a deterministic comparison tree and a leaf $\ell$ reached with positive
probability under a uniform order. Let $E_\ell$ be the set of orders that
reach $\ell$, let $a$ be its output label, and let $C$ be the connected
component of $a$ in the graph whose edges are the distinct pairs compared on
the path to $\ell$. Put $k=|C|$. Conditional on $E_\ell$, the set of ranks
occupied by $C$ is a uniform $k$-subset of $[n]$, independent of the order
within $C$, and
\begin{equation}\label{eq:posterior}
\Pr[x_a=t\mid E_\ell]=\sum_{j=1}^k\alpha_j\,
\frac{\binom{t-1}{j-1}\binom{n-t}{k-j}}{\binom nk},
\end{equation}
where $\alpha_j\ge0$ is the fraction of linear extensions of the order
recorded on $C$ in which $a$ has the $j$-th smallest rank within $C$, and
$\sum_j\alpha_j=1$.
\end{lemma}
\begin{proof}
Each comparison of distinct labels on the path to $\ell$ records one strict
inequality. A self-comparison imposes no constraint. An order
reaches $\ell$ exactly when it satisfies all recorded inequalities, because
the tree is deterministic and every answer on the path is prescribed. So
$E_\ell$ is the set of linear extensions of the recorded partial order. Let
$D$ be the complement of $C$. No recorded inequality joins $C$ and $D$, and
the transitive closure adds none, since a directed path between them would
be an undirected path in the comparison graph. Fix a $k$-subset $S$ of ranks
and a linear extension $\pi_C$ of the order on $C$. Place $\pi_C$ on the ranks $S$ in increasing order. For each linear extension
$\pi_D$ of the order on $D$, place $\pi_D$ on the complementary ranks. This
gives exactly one order in $E_\ell$ with rank set $S$ on $C$ and internal
order $\pi_C$, and every such order arises once. The count is the number $e_D$ of linear extensions of the order
on $D$, which does not depend on $S$ or $\pi_C$. The joint law of rank set
and internal order is therefore uniform, which proves independence. If $a$
has the $j$-th smallest rank within $C$ under $\pi_C$, then $x_a=t$ exactly
when $S$ contains $t$, $j-1$ ranks below $t$ and $k-j$ ranks above it, which
gives~\eqref{eq:posterior}.
\end{proof}

\begin{proof}[Proof of Theorem~\ref{thm:classical}]
\emph{Attainment.} The output rank is the minimum of a uniform $(q+1)$-subset
of ranks. It misses the top $r$ exactly when all $q+1$ ranks exceed $r$, which
has probability $\binom{n-r}{q+1}/\binom n{q+1}$. The output rank equals $t$ exactly when
the subset contains $t$ and $q$ ranks above $t$, which
gives~\eqref{eq:classical-law}. Because the subset is chosen uniformly, these
probabilities hold on every order.

\emph{Upper bound.} By Proposition~\ref{prop:twirl} it suffices to bound the
success averaged over a uniform order. Fixing the private randomness turns
an algorithm into a deterministic comparison tree with at most $q$
comparisons on every branch, and the average success of the algorithm is
the average over seeds of the average success of these trees. Fix a tree and
a leaf $\ell$ of positive probability. Its comparison graph has at most $q$
distinct edges, and a connected graph on $k$ vertices has at least $k-1$
edges, so the component of the output has $k\le q+1$ vertices. Let
$V_1,\ldots,V_{q+1}$ be a uniformly random sequence of distinct ranks. The
set $\{V_1,\ldots,V_k\}$ is a uniform $k$-subset, and for every $j\le k$ its
$j$-th smallest element is at least $\min(V_1,\ldots,V_k)$, which is at least
$\min(V_1,\ldots,V_{q+1})$. Each term of~\eqref{eq:posterior} is the law of
the $j$-th smallest element of a uniform $k$-subset, so
\[
\Pr[x_a\le r\mid E_\ell]\le\Pr[\min(V_1,\ldots,V_{q+1})\le r]
=1-\frac{\binom{n-r}{q+1}}{\binom n{q+1}}.
\]
Averaging over leaves and seeds gives the bound. Repeated or reversed
comparisons add no distinct edge, and a self-comparison adds no edge between
distinct labels. Leaves with contradictory answers have probability zero.
An output that was never compared has $k=1$ and rank law $1/n$.
\end{proof}

\begin{corollary}\label{cor:classical}
(i) For $n\ge q+1$ and $1\le r\le n-q$, the equivalent blind-pick size of
$q$ exact classical comparisons is exactly $s_{\rm eq}^{C}(n,r)=q+1$.
(ii) The rank law of every symmetrized classical algorithm with budget $q$
is, on the rank grid, a polynomial in the rank of degree at most
$\min(q,n-1)$.
\end{corollary}
\begin{proof}
Part (i) is~\eqref{eq:classical}, because $q+1\le n-r+1$ lies in the range
where $B(n,r,\cdot)$ is strictly increasing. For (ii), the symmetrized rank
law is the average of~\eqref{eq:posterior} over leaves weighted by their
probabilities, and each term of~\eqref{eq:posterior} has degree $k-1\le q$
in $t$. Every function on $n$ points is a polynomial of degree at most $n-1$.
\end{proof}

Theorem~\ref{thm:classical} determines the best comparison-based procedure
when comparisons are exact and one design is selected. No adaptive choice of
pairs improves on a blind pick of $q+1$ designs followed by a knockout,
although $q$ comparisons can involve up to $2q$ designs. Any procedure that returns the best of the same $q+1$ designs with $q$ comparisons
attains the same value, whether it compares the current winner with the next design, as
in the proof, or plays a bracket. The blind-pick formula itself is
classical~\cite{HoZhaoJia2007,LauHo1997}, and the winner-rank law of
tournament selection without replacement gives the same
expression~\cite{XieEtAl2008}.

For the upper bound, conditioning on a leaf of a comparison tree
with shuffled inputs appears in work on approximate selection with
unreliable comparisons~\cite{HuangLiuRutschmann2023}, and lower bounds there
and in~\cite{LeucciLiu2022} fix the asymptotic order of the number of
comparisons. The proof of Theorem~\ref{thm:classical} is elementary. It uses
leaf conditioning, a count of linear extensions and a stochastic comparison
of order statistics. The theorem gives the optimum over all adaptive procedures at every
budget $q\le n-1$ and every good-enough set.
For $q\ge n-1$ a knockout of all $n$ designs returns the best one, so
$C_q(n,r)=1$ for every $r\ge1$. In the
scenario-by-scenario comparison of selection rules proposed
in~\cite{HoZhaoJia2007}, Theorem~\ref{thm:classical} gives the best rule for
the scenario of exact comparisons and one selected design.
Earlier decision-tree bounds treat randomized decision problems~\cite{ManberTompa1985}
and verification that a designated element is the maximum~\cite{Snir1985},
rather than selection of one design at a fixed budget.

\section{Rank laws of quantum algorithms}\label{sec:degree}

The classical degree bound limits how output probabilities can vary with
rank. We now derive the corresponding bound for coherent comparisons and
show that it can be attained. The factor two in the degree bound occurs
because a probability is quadratic in the amplitudes.

\begin{theorem}\label{thm:degree}
Let $n\ge2$ and let an admissible quantum algorithm use at most $q$ calls
to~\eqref{eq:oracle} on every branch. After the symmetrization of
Proposition~\ref{prop:twirl}, its rank law $w_t=\Pr[x_I=t]$, $t\in[n]$, is
independent of the order and satisfies $w_t=p(t)$ for a polynomial $p$ of
degree at most $\min(2q,n-1)$. The degree bound is attained for every $q$
and $n$.
\end{theorem}
\begin{proof}
Write $\sgn_{ab}(x)=\sgn(x_a-x_b)=(-1)^{\ind{x_a<x_b}}$ for one orientation of each unordered
pair. In the $\ket\pm$ basis of the response bit, each call acts on the
address sector $(a,b)$ as $\Pi_++\sgn_{ab}\Pi_-$, on the reversed sector $(b,a)$
as $\Pi_+-\sgn_{ab}\Pi_-$, and on a self-address as the identity. Every unnormalized Kraus
trajectory with $k$ calls therefore has amplitudes of degree at most $k$ in
the signs, and input-independent operations do not raise the degree. The
unconditional probability of each output label is a sum over trajectories,
Kraus indices and initial mixtures, so it has degree at most $2q$. Using
$\sgn_{ab}^2=1$, the probability of output $a$ on order $x$ is
\begin{equation}\label{eq:signs}
p_a(x)=\sum_{E:\,|E|\le2q}c_{a,E}\prod_{\{u,v\}\in E}\sgn_{uv}(x),
\end{equation}
which is the polynomial method~\cite{BealsEtAl2001} written in comparison
signs.

Regard a monomial in~\eqref{eq:signs} as a graph $G$ with edge set $E$, add
the output label $a$ as a vertex, and let the component of $a$ contain $k$
other vertices. Let $\mu_j$ be the average of the sign product of a detached
component $C_j$ over its internal orders, and let $\beta_h$ be the average
sign product of the root component over the orders of its vertices in which
exactly $h$ of the other $k$ vertices have a smaller rank than $a$, with
$\beta_0=1$ when $k=0$.
For a uniform order $X$,
\begin{equation}\label{eq:rooted}
\E\Bigl[\prod_{\{u,v\}\in E}\sgn_{uv}(X)\Bigm|X_a=t\Bigr]
=\Bigl(\prod_{j\ge1}\mu_j\Bigr)\,
\frac{\sum_{h=0}^k\beta_h\binom{t-1}{h}\binom{n-t}{k-h}}{\binom{n-1}{k}}.
\end{equation}
To derive~\eqref{eq:rooted}, we group the orders with $X_a=t$ by the rank
sets assigned to the graph components. Within each group, internal
assignments are independent and uniform, with rank $t$ fixed at $a$.
Each detached component contributes its average sign product $\mu_j$.
The root component's other $k$ ranks form a uniform subset of
$[n]\setminus\{t\}$, so the probability that exactly $h$ fall below $t$ is
the displayed hypergeometric weight. The right side of~\eqref{eq:rooted} has
degree at most $k\le\min(|E|,n-1)$ in $t$. The symmetrized law is
$w_t=n^{-1}\sum_a\E[p_a(X)\mid X_a=t]$, which proves the bound.

For attainment, pick a uniform root label $a$ and a uniform $d$-subset of
the other labels, with $d=\min(2q,n-1)$. Pair these labels. For a pair
$\{u,v\}$, prepare a selector qubit in $\ket+$ and the response in $\ket-$,
encode the selector into the addresses $(a,u)$ and $(a,v)$, call the oracle
once and undo the encoding. Apply a Hadamard gate to the selector and measure it to determine the
product $\sgn_{au}\sgn_{av}$ exactly~\cite{CleveEtAl1998}. If $d$ is odd, read the
last sign with one more call. Output $a$ when the product of all $d$ signs is
$+1$, and otherwise output a uniform label other than $a$. With
\[
\phi_d(t)=\frac{\sum_{h=0}^d(-1)^{d-h}\binom{t-1}{h}\binom{n-t}{d-h}}{\binom{n-1}{d}},
\qquad
\bar\phi_d=\begin{cases}1/(d+1),&d\text{ even},\\0,&d\text{ odd},\end{cases}
\]
the rank law is
$w_t=n^{-1}+\bigl(\phi_d(t)-\bar\phi_d\bigr)/\bigl(2(n-1)\bigr)$. For $d\ge1$
its coefficient of $t^d$ is $2^d/\bigl(2(n-1)\,d!\binom{n-1}{d}\bigr)>0$, and
for $q=0$ the law is uniform, of degree zero. The algorithm uses
$\lceil d/2\rceil\le q$ calls.
\end{proof}

At $n=3$ and $q=1$ the construction gives the rank law
$(1/2,\,0,\,1/2)$, which has degree two and so cannot come from one
classical comparison. By Corollary~\ref{cor:classical} and
Theorem~\ref{thm:degree}, one quantum comparison permits rank laws of degree
two, whereas one classical comparison permits degree one. The larger degree
is necessary for a quantum gain at one comparison. A nonnegative rank law of
degree at most one is affine in the rank and sums to one, so it is a mixture
of the law~\eqref{eq:classical-law} with $q=1$ and its reversal, and its
top-$r$ mass is at most $C_1(n,r)$ for every $r$. For $r=1$ this reads
$n\,w_1\le2=n\,C_1(n,1)$.

However, the degree bound does not determine the optimal alignment
probability. One quantum comparison and two classical comparisons share
the degree bound two and have different values (Section~\ref{sec:one}).
Determining the alignment optimum therefore requires optimization over
rank laws that admissible quantum algorithms can realize.

\section{The exact optimum with one quantum comparison}\label{sec:one}

To determine this optimum, we need an attaining algorithm and an upper
bound that applies to every admissible one-comparison strategy. The
following theorem supplies both and expresses their common value through
an explicit cubic.

\begin{theorem}\label{thm:one}
$Q_1(2,1)=1$. For $n\ge3$ and $1\le r<n$, put $m=n-1$, $k=n-2$ and $h=n-2r$,
and let $\zeta$ be the smallest positive root of
\begin{equation}\label{eq:cubic}
F_{n,r}(z)=2nhz^3+3(nk-h^2)z^2-4(2n-3)hz-6mk.
\end{equation}
Then
\begin{equation}\label{eq:one}
Q_1(n,r)=\frac rn+\frac{r(n-r)}{n(n-1)}\,\zeta .
\end{equation}
For $r<n/2$ the positive root is unique and lies in $(1,3)$. For $r=n/2$ it
is $\sqrt{2(n-1)/n}$. For $n/2<r<n-1$ it is the unique root in
$\bigl(1,(n-2)/(2r-n)\bigr)$, and for $r=n-1$ it is $1$. An explicit initial state and output measurement attain this value for every
unknown order.
\end{theorem}

\begin{corollary}\label{cor:one}
For $n\ge2$,
\[
Q_1(n,1)=\frac{3(n-1)+\sqrt{13n^2-18n+9}}{2n^2},
\qquad\text{and}\qquad
Q_1(n,n/2)=\frac12+\sqrt{\frac{n}{8(n-1)}}
\]
for even $n$. The equivalent blind-pick size $n\,Q_1(n,1)$ increases with $n$
toward $(3+\sqrt{13})/2\approx3.3028$, and for $n\ge3$,
$n\,Q_1(n,1)>3=n\,C_2(n,1)$ exactly when $n>9$, with equality at $n=9$.
\end{corollary}
\begin{proof}
The closed forms are proved in Appendix~\ref{app:one}. With $u=1/n$,
$n\,Q_1(n,1)=\bigl(3-3u+\sqrt{13-18u+9u^2}\bigr)/2$. Its derivative in $u$ is
$-3/2+(9u-9)/\bigl(2\sqrt{13-18u+9u^2}\bigr)<0$ for $0<u\le1/2$, so the size
increases with $n$, and its value at $u=0$ is $(3+\sqrt{13})/2$. For the
comparison with $C_2(n,1)=3/n$ from Theorem~\ref{thm:classical},
$3(n-1)+\sqrt{13n^2-18n+9}>6n$ is equivalent to $13n^2-18n+9>9(n+1)^2$, that
is to $4n(n-9)>0$.
\end{proof}

The classical value with one comparison is $C_1(n,r)=r/n+r(n-r)/(n(n-1))$ by
Theorem~\ref{thm:classical}. It is~\eqref{eq:one} with $\zeta$ replaced by
$1$. Since $\zeta\ge1$, with equality only at $r=n-1$, one quantum comparison
is strictly better than one classical comparison for $1\le r\le n-2$.

The same formula gives the limiting alignment probability when the good-enough
set occupies a fixed fraction of the designs.
\begin{corollary}\label{cor:fraction}
Fix $0<\varrho<1$ and put $r_n=\lfloor\varrho n\rfloor$. Then
\[
\lim_{n\to\infty}Q_1(n,r_n)=F_1(\varrho)
=\varrho+\varrho(1-\varrho)\zeta(\varrho),
\]
where $\zeta(\varrho)$ is the smallest positive root of
\begin{equation}\label{eq:fraction-cubic}
H_{\varrho}(z)=2(1-2\varrho)z^3+12\varrho(1-\varrho)z^2
-8(1-2\varrho)z-6=0.
\end{equation}
In particular, $F_1(1/2)=1/2+\sqrt2/4$ and
$\lim_{\varrho\downarrow0}F_1(\varrho)/\varrho=(3+\sqrt{13})/2$.
\end{corollary}
\begin{proof}
Write $r_n=\varrho n+e_n$, where $-1<e_n\le0$. Dividing
$F_{n,r_n}$ in~\eqref{eq:cubic} by $n^2$ gives coefficientwise convergence
to $H_{\varrho}$, since every remaining term has a bounded coefficient
times $1/n$ or $1/n^2$. Let $\zeta_n$ be the root selected in
Theorem~\ref{thm:one}. If $\varrho<1/2$, then $1<\zeta_n<3$ for all
sufficiently large $n$. The signs of the coefficients of $H_{\varrho}$
are $+,+,-,-$, so it has exactly one positive root. Every subsequential
limit of $\zeta_n$ is a root of $H_{\varrho}$ in $[1,3]$, and thus the
whole sequence converges to that positive root. At $\varrho=1/2$, the
same boundedness holds for both parities of $n$, and the limiting
polynomial is $3z^2-6$. Its only root in $[1,3]$ is $\sqrt2$.

For $\varrho>1/2$, put $a=2\varrho-1\in(0,1)$. Theorem~\ref{thm:one}
gives $1<\zeta_n<(n-2)/(2r_n-n)$ for all sufficiently large $n$.
The upper endpoint tends to $1/a$. Moreover,
$H_{\varrho}(1)=-12(1-\varrho)^2<0$ and
$H_{\varrho}(1/a)=a^{-2}-1>0$. The negative leading coefficient gives
a second positive root above $1/a$, and Descartes' rule permits at most
two. There is therefore exactly one root in $[1,1/a]$, and it is the
smallest positive root. Boundedness and coefficient convergence again
give $\zeta_n\to\zeta(\varrho)$. Substitution in~\eqref{eq:one} proves
the asserted limit for every $0<\varrho<1$.

The value at $1/2$ follows from $\zeta(1/2)=\sqrt2$. As
$\varrho\downarrow0$, the selected root remains in $[1,3]$, and
$H_0(z)=2(z+1)(z^2-z-3)$ has only one root in this interval.
Consequently $\zeta(\varrho)\to(1+\sqrt{13})/2$, which gives the stated
small-fraction limit.
\end{proof}

The proof of Theorem~\ref{thm:one} in Appendix~\ref{app:one} uses the same moment matrices for the
upper bound and the attaining construction. The state after one
call is a linear image of the feature vector
$\phi_x=\ket0+\sum_{i<j}\sgn(x_i-x_j)\,e_i\wedge e_j$, so the success of any
algorithm is $\sum_i\tr(M_iB_i)$ for target-conditioned moment matrices
$M_i$ and positive semidefinite matrices $B_i$ whose sum is a nonnegative diagonal matrix of
trace one. This is the
one-query case of the semidefinite characterization of quantum query
algorithms~\cite{BarnumSaksSzegedy2003}, reduced by label symmetry. The moment matrices
$M_i$ decompose into explicit $2\times2$ blocks under label
permutations. A dual matrix $\mathcal L\succeq M_i$ with constant diagonal
then gives an upper bound for every algorithm, and an explicit covariant state and measurement
attain the bound. The value $\zeta$ is the root at which the diagonal of
$\mathcal L$ becomes constant.

\section{Comparison in ordinal-optimization terms}\label{sec:oo}

We use the equivalent blind-pick size to compare the exact values in OO
terms. We first consider selection of the best design, then examine how
the comparison changes when more designs are accepted as good enough.
Table~\ref{tab:seq} reports the sizes $n\,C_q(n,1)$
and $n\,Q_1(n,1)$ for selecting the best of $n$ designs. The classical values
are exact for every $n\ge q+1$ (Theorem~\ref{thm:classical}). The quantum
values follow from Corollary~\ref{cor:one} and are rounded to four decimals.
The quantum entries at $n=2$ and $n=9$ are displayed exactly. The value
at $n=3$ is $1+\sqrt2$. For one classical comparison the size is $2$
for every $n\ge2$ (Corollary~\ref{cor:classical}). For one quantum comparison
it increases from $2$ toward about $3.30$ as $n$ grows
(Corollary~\ref{cor:one}). The values coincide at $n=2$, because one comparison
orders two designs. One quantum comparison has a larger equivalent
blind-pick size than two classical comparisons exactly when $n\ge10$. The classical
entry with two comparisons is omitted at $n=2$, where one comparison already identifies
the best design.

\begin{table}[htbp]\centering
\caption{Equivalent blind-pick sizes for selecting the best design.}
\label{tab:seq}
\small
\begin{tabular}{cccc}\toprule
$n$ & classical, $q=1$ & classical, $q=2$ & quantum, $q=1$\\\midrule
2 & 2 & & 2\\
3 & 2 & 3 & 2.4142\\
4 & 2 & 3 & 2.6302\\
7 & 2 & 3 & 2.9145\\
9 & 2 & 3 & 3\\
10 & 2 & 3 & 3.0300\\
20 & 2 & 3 & 3.1659\\
100 & 2 & 3 & 3.2753\\
$\infty$ & 2 & 3 & 3.3028\\\bottomrule
\end{tabular}
\end{table}

For a fixed pool of seven designs, we next vary the good-enough target.
Table~\ref{tab:seven} reports alignment probabilities for the top $r$,
for $1\le r\le6$. It also gives the equivalent blind-pick
size $s_{\rm eq}^{Q}(7,r)$ of one quantum comparison. The blind pick
selects one design. Classical values follow from Theorem~\ref{thm:classical}
and quantum values from Theorem~\ref{thm:one}. Decimal values are rounded
to four places. The equivalent blind-pick size of one classical comparison
is $2$ for every $r$ in the table (Corollary~\ref{cor:classical}).
For $r\le5$, the probability with one quantum
comparison is strictly between those with one and two classical comparisons.
An alignment probability of $0.8$ for the best three of seven requires two
classical comparisons, since $C_1(7,3)=5/7$, whereas one quantum comparison
suffices. The same holds for an alignment probability of $0.9$ for the best
four, since $C_1(7,4)=6/7$.

\begin{table}[htbp]\centering
\caption{Selection performance for seven designs.}
\label{tab:seven}
\small
\begin{tabular}{cccccc}\toprule
& blind & classical & quantum & classical & size\\
$r$ & pick & $q=1$ & $q=1$ & $q=2$ & $s_{\rm eq}^{Q}(7,r)$\\\midrule
1 & 0.1429 & 0.2857 & 0.4164 & 0.4286 & 2.9145\\
2 & 0.2857 & 0.5238 & 0.6769 & 0.7143 & 2.7826\\
3 & 0.4286 & 0.7143 & 0.8314 & 0.8857 & 2.6151\\
4 & 0.5714 & 0.8571 & 0.9199 & 0.9714 & 2.4157\\
5 & 0.7143 & 0.9524 & 0.9705 & 1 & 2.1953\\
6 & 0.8571 & 1 & 1 & 1 & 2\\\bottomrule
\end{tabular}
\end{table}

In OO applications the good-enough set is usually a small fraction of many
designs. To examine this regime, Table~\ref{tab:thousand} reports alignment
probabilities for the top $r$ of $n=1000$ designs and the equivalent
blind-pick size $s_{\rm eq}^{Q}(1000,r)$ of one quantum comparison.
Values are rounded to four decimals. The equivalent blind-pick size of
one quantum comparison is about $3.27$ for the top
$5$ percent, $3.23$ for the top $10$ percent and $3.14$ for the top
$20$ percent. Each exceeds the size $3$ of two classical comparisons.
For the top half the quantum size is $2.77$, below $3$. In the cases computed,
the equivalent size decreases as the good-enough set grows, so the quantum
gain in these units is largest for selecting the best design.

\begin{table}[htbp]\centering
\caption{Selection performance for $1000$ designs.}
\label{tab:thousand}
\small
\begin{tabular}{cccccc}\toprule
& good-enough & classical & quantum & classical & size\\
$r$ & fraction & $q=1$ & $q=1$ & $q=2$ & $s_{\rm eq}^{Q}(1000,r)$\\\midrule
1 & 0.1\% & 0.0020 & 0.0033 & 0.0030 & 3.3000\\
10 & 1\% & 0.0199 & 0.0326 & 0.0297 & 3.2939\\
50 & 5\% & 0.0975 & 0.1544 & 0.1428 & 3.2655\\
100 & 10\% & 0.1901 & 0.2885 & 0.2712 & 3.2265\\
200 & 20\% & 0.3602 & 0.5038 & 0.4884 & 3.1364\\
500 & 50\% & 0.7503 & 0.8537 & 0.8754 & 2.7698\\\bottomrule
\end{tabular}
\end{table}

For a fixed good-enough fraction, Corollary~\ref{cor:fraction} gives a
comparison that no longer depends on the number of designs.
The classical limit is $\lim_{n\to\infty}C_1(n,\lfloor\varrho n\rfloor)
=1-(1-\varrho)^2$, by Theorem~\ref{thm:classical}.
Table~\ref{tab:fraction} compares this value with $F_1(\varrho)$.
At a good-enough fraction of $5$ percent, the limiting quantum success
probability is about $0.1543$, close to the finite value $0.1544$ for
$1000$ designs in Table~\ref{tab:thousand}.
\begin{table}[ht]
\centering
\caption{Fixed-fraction limits.}\label{tab:fraction}
\begin{tabular}{rrr}\toprule
Good-enough fraction & Classical & Quantum\\\midrule
$1\%$ & 0.019900 & 0.032583\\
$5\%$ & 0.097500 & 0.154311\\
$10\%$ & 0.190000 & 0.288343\\
$20\%$ & 0.360000 & 0.503566\\
$50\%$ & 0.750000 & 0.853553\\\bottomrule
\end{tabular}
\end{table}

With exact comparisons and one selected design, no adaptive choice of pairs
improves on a blind pick of $q+1$ designs followed by a knockout. The budget needed for
a target probability therefore follows from~\eqref{eq:blind}. If the exact
comparison oracle can be queried coherently, the equivalent blind-pick size
of one quantum comparison for selecting the best design approaches about
$1.65$ times the size of one classical comparison as $n$ grows. In the cases
computed this ratio decreases when the good-enough set is a larger fraction
of the designs. It is about $1.63$ for the top $5$ percent of $1000$ designs
and $1.38$ for the top half. These ratios measure gains per comparison
under the exact oracle model. Their computational significance depends on
how the comparisons are implemented.

\section{Discussion and future work}\label{sec:program}

The results above apply to the instance of QOO with exact comparisons and
one selected design. The following questions consider other instances of
Definition~\ref{def:qoo} with larger quantum budgets, comparisons inferred
from evaluations and selection of multiple designs. They also address the cost
of implementing coherent comparisons and checking feasibility. These
changes affect the information available, the success criterion or
the cost assigned to a query, so each extension requires a common
information and cost model for classical and quantum procedures.

\subsection{Larger comparison budgets}\label{sec:program-budgets}
The two-comparison problem is studied in a companion paper~\cite{WuTwoComparisons2026},
which gives a complete semidefinite characterization of attainable success.
Computer-assisted proofs determine the optimum for seven designs and every
good-enough set, and for selecting the best of eight designs.
In particular, $Q_2(7,1)\approx0.7166$ and $Q_2(8,1)\approx0.6586$.
For best-design selection, $nQ_2(n,1)$ is nondecreasing and bounded for
$n\ge7$, with a limit in $(7.43710307,7.43815019]$.
This limit is strictly greater than that of the explicit algorithm family
constructed in the companion paper. Two quantum comparisons beat six
classical comparisons exactly when $n\ge46$.
The values for general $n$ and $r$ and the precise scaled limit remain open.
Although Theorem~\ref{thm:degree} restricts the possible rank laws, it does
not by itself characterize attainability. The two-comparison
characterization resolves this issue at that budget, while explicit
reductions and optimal values for higher budgets remain to be determined.

For every budget $q\ge1$, the search lower bound~\cite{Zalka1999,BennettEtAl1997} gives $n\,Q_q(n,1)=O(q^2)$ by a
reduction. Rank a single marked item first and order the others by label. A
comparison of distinct labels $i$ and $j$ is then true exactly when $i$ is
marked, or when neither is marked and $i<j$, which four search queries
compute and uncompute. Minimum finding~\cite{DurrHoyer1996} reaches the minimum of a set of designs after an expected number of comparisons of the
order of the square root of their number. Stopped at a constant multiple of
that expectation, it has found the minimum with probability bounded below by a constant,
by Markov's inequality, so applied to a blind pick of order $q^2$ designs it achieves $n\,Q_q(n,1)$ of order $q^2$
when $n$ is at least of that order. The exact values and the sharp constants in these
bounds are open.

These search bounds are for selecting the best design. For OO, larger
budgets should also be studied when the good-enough set
remains a fixed fraction of all designs as their number grows. The
one-comparison values in Section~\ref{sec:oo} show that the relative gain in
blind-pick units depends on this fraction. Determining that dependence at
larger budgets would show when accepting more designs as good enough
changes the benefit of quantum comparisons.

\subsection{Evaluation and selection models}\label{sec:program-models}
The preceding questions keep comparisons exact. In simulation-based OO,
however, comparisons are inferred from observations.
Their error probabilities can depend on the performance gaps, so the
ordered performance curve affects the alignment
probability~\cite{LauHo1997}. Asymptotic query counts for best-arm
identification depend on reward
gaps~\cite{WangYouLiChilds2021,BuchholzKublerScholkopf2025}. The corresponding
question here is the largest alignment probability at a fixed budget,
under a specified model of comparison errors.

When comparisons are computed from quantum estimates of design
means~\cite{Montanaro2015,Hamoudi2021}, the budget counts calls to
the underlying evaluation oracle. The estimation precision and the number
of comparisons then have to be chosen together. The problem is to allocate
that budget so as to maximize alignment probability, a quantum counterpart
of the allocation question addressed by
OCBA~\cite{ChenLinYucesanChick2000}.

An evaluation model can also change what one oracle call answers. A
membership query tests whether a single design belongs to the good-enough
set instead of comparing two designs. The fixed-budget values of this model
can be studied when the number of good designs is
unknown~\cite{BoyerBrassardHoyerTapp1998,ChakrabortyEtAl2005}.

Beyond changes to the evaluation oracle, the output criterion can also
change. This paper selects one design and
uses alignment level one. Selecting a set with a prescribed alignment
level introduces the selected-set size as a further parameter. A
fixed-confidence quantum algorithm already identifies a given number of
best arms as a set through a reward oracle~\cite{WangEtAl2025QCE}. The
problem for QOO is to determine the best alignment probability under a
fixed comparison budget and the corresponding exact classical benchmark.

\subsection{Computational cost and constrained selection}\label{sec:program-cost}
The evaluation-oracle budget above counts queries rather than total
computation. Computational cost also includes the implementation of a coherent
comparison, its workspace and the error-correction overhead~\cite{BabbushEtAl2021}.
The practical question is whether the gain
in alignment probability justifies that cost for a specified evaluation
model and target probability.

For constrained selection, the cost model also includes feasibility
checks. In constrained OO for unit commitment, the selection criterion
combines feasibility with sufficiently low cost~\cite{WuShahidehpour2014}.
The corresponding QOO formulation specifies how feasibility is checked,
what information a query returns and how those checks are charged. A performance comparison
alone does not supply this information.

Power system screening with deterministic, expensive evaluations on a
fixed set of scenarios can provide a first test of this cost model.
Classical and quantum procedures would use the same candidates,
feasibility rules and target alignment probability. Measuring evaluation,
comparison and selection costs would assess the total cost of reaching
that target, alongside the alignment probability per query.

\section{Conclusion}\label{sec:conclusion}

We formulated quantum ordinal optimization as a class of problems that
maximize worst-case alignment probability at a fixed oracle budget. We
studied the instance with exact pairwise comparisons, selected-set size one
and alignment level one. For exact classical comparisons and a budget smaller than the number of
designs, uniform sampling without replacement followed by a knockout is
optimal. The sample contains one more design than the budget. With quantum comparisons, the rank law
averaged over a random order can have twice the classical degree.
We determined the optimum with one comparison from the root of a cubic.
These values quantify the benefit of coherent comparison access
relative to the optimal classical procedure, and known search bounds show
that the quantum-to-classical ratio in blind-pick units grows at most in
proportion to the budget for the best design. To determine how that benefit
changes beyond the present model, we identify
questions about larger budgets, noisy and estimated comparisons, and the
computational cost of each comparison as future work. 

\section*{Author contribution and use of AI tools}
The author used AI tools to assist with literature research, mathematical analysis, code development, and manuscript preparation. The author assumes
full responsibility for the content of the article.

\appendix
\section{Proof of Theorem~\ref{thm:one}}\label{app:one}

We first isolate the root used in the theorem and establish positivity of
the parameters. We then compute the target-conditioned moments, use them
to bound success for every admissible algorithm and construct a state and
measurement that attain the bound.

Throughout, $n\ge3$ and $1\le r<n$, with $m=n-1$, $k=n-2$, $p=r/n$, $\upsilon=n-r$
and $h=n-2r$. For $n=2$, one classical comparison suffices.

\subsection{Root isolation and positivity}
The constraints give $|h|\le k$ and $nk-h^2\ge2k>0$. Direct substitution gives
\begin{equation}\label{eq:Fone}
F_{n,r}(1)=-3(h+k)^2=-12(n-r-1)^2 .
\end{equation}
If $h>0$, the coefficients of $F_{n,r}$ have signs $+,+,-,-$ in descending
degree, so Descartes' rule gives exactly one positive root.
Equation~\eqref{eq:Fone} implies that the root exceeds $1$. Since $h\le n-2$,
$F_{n,r}(3)=(21n+6)k+h(30n+36-27h)>0$, so the root is less than $3$.
If $h=0$, the polynomial is $3nkz^2-6mk$, which gives the stated radical.
If $h<0$ and $r<n-1$, write $a=-h$, so $0<a<k$. Then
$F_{n,r}(k/a)=nk(k^2/a^2-1)>0$, so there is a root between $1$ and $k/a$. The
negative leading coefficient gives another positive root above $k/a$, and
Descartes' rule allows at most two, so the first is the smallest. At
$r=n-1$, $F_{n,n-1}(z)=-2k(z-1)\bigl[nz^2+(n-3)z-3m\bigr]$, and the bracket
has one positive root, which exceeds $1$ because the bracket equals $-n$ at
$z=1$. The smallest positive root is then $1$.

Define
\[
\kappa=3k+2h\zeta,\qquad \rho=\frac{\zeta}{\sqrt m},\qquad
\gamma=\frac{h\zeta}{\sqrt m\,\kappa},\qquad R=1-\frac{k\gamma^2}{2}.
\]
For $h>0$, $\kappa\ge3k$ and $0<\gamma<1/(2\sqrt m)$. For $h=0$, $\gamma=0$.
For $h<0$, the selected root satisfies $\zeta\le k/(-h)$, so $\kappa\ge k$ and
$|\gamma|\le1/\sqrt m$, which gives $R\ge1-k/(2m)=n/(2m)>0$. So $\kappa>0$ and
$R>0$ in all cases.

\subsection{Target-conditioned moments}
Use the real feature space $\mathcal H=\R\ket0\oplus\Lambda^2\R^n$ with
$\phi_x=\ket0+\sum_{i<j}\sgn(x_i-x_j)\,e_i\wedge e_j$, where the oriented
edges $e_i\wedge e_j$, $i<j$, are orthonormal. Let
$\omega=n^{-1/2}\sum_ie_i$, and let $\mathcal H_0=\R\ket0$,
$\mathcal V=\{v\wedge\omega:v\perp\omega\}$ and
$\mathcal C=\Lambda^2(\omega^\perp)$, of dimensions $1$, $m$ and $mk/2$, with
orthogonal projectors $P_0$, $P_V$ and $P_C$. For output label $1$ put
$b=\sqrt{n/m}\,(e_1-n^{-1}\sum_ie_i)$ and
$W=\{w\in\R^n:w_1=0,\ \sum_{j\ge2}w_j=0\}$, of dimension $k$, with an
orthonormal basis $f_1,\ldots,f_k$. Use the decomposition
\begin{equation}\label{eq:local}
\mathcal H=\operatorname{span}\{\ket0,b\wedge\omega\}\oplus
\bigoplus_{a=1}^k\operatorname{span}\{f_a\wedge\omega,\ b\wedge f_a\}
\oplus\Lambda^2W .
\end{equation}
For a uniform order define $M_i=\E_x\bigl[\ind{x_i\le r}\phi_x\phi_x^{\mathsf T}\bigr]$.

\begin{lemma}\label{lem:moments}
Put $\nu=p(n+1)/3$, $\delta=p\upsilon h/(3mk)$, $\tau=p/3$,
$T=\nu+2k\delta=(p/m)\bigl(\upsilon^2+(r^2-1)/3\bigr)$, $S=\nu-2\delta$,
$\theta=p\upsilon/\sqrt m$ and $\eta=\sqrt m\,\delta$. In the decomposition~\eqref{eq:local},
\[
M_1=\begin{pmatrix}p&-\theta\\-\theta&T\end{pmatrix}\oplus
\bigoplus_{a=1}^k\begin{pmatrix}S&\eta\\\eta&\tau\end{pmatrix}\oplus\tau I_{k(k-1)/2}.
\]
\end{lemma}
\begin{proof}
Condition on $x_1=\ell$ and put $L_\ell=2\ell-n-1$. The signs on edges at
label $1$ sum to $L_\ell$, and $b\wedge\omega=m^{-1/2}\sum_{j\ge2}e_1\wedge e_j$,
so the coordinate along $b\wedge\omega$ is $L_\ell/\sqrt m$. The sums
$n^{-1}\sum_{\ell\le r}L_\ell=-p\upsilon$ and
$n^{-1}\sum_{\ell\le r}L_\ell^2=p\bigl(\upsilon^2+(r^2-1)/3\bigr)$ give the first block.
For a unit vector $w\in W$ use $g_w=e_1\wedge w$ and $y_w=w\wedge\omega_o$
with $\omega_o=m^{-1/2}\sum_{j\ge2}e_j$. The star signs at label $1$ contain
$\ell-1$ plus signs and $n-\ell$ minus signs, and exchangeability gives
$A_\ell:=\E[\langle g_w,\phi_x\rangle^2\mid x_1=\ell]=4(\ell-1)(n-\ell)/(mk)$.
For a uniform order on $m$ labels the oriented-edge covariance is
$(I+D_m^{\mathsf T}D_m)/3$, where $D_m$ has edge columns $e_i-e_j$. The
relative order of the other labels is uniform given $x_1=\ell$, so
$\E[\langle y_w,\phi_x\rangle^2\mid x_1=\ell]=n/3$. For the cross term, let $Y_j$ be the relative rank of a label $j\ne1$ among
the $m$ labels other than $1$. This rank is uniform on $[m]$. Then
$\E[\sgn(\ell-x_j)(2Y_j-m-1)]=-2(\ell-1)(m-\ell+1)/m$, and since
rank scores and the coordinates of $w$ sum to zero,
$\E[\langle g_w,\phi_x\rangle\langle y_w,\phi_x\rangle\mid x_1=\ell]=-\sqrt m\,A_\ell/2$.
Averaging over $\ell\le r$ gives
$A_0=n^{-1}\sum_{\ell\le r}A_\ell=2p(r-1)(3n-2r-2)/(3mk)$ and the block
$\left(\begin{smallmatrix}A_0&-\sqrt m A_0/2\\-\sqrt m A_0/2&pn/3\end{smallmatrix}\right)$.
The rotation $w\wedge\omega=(-g_w+\sqrt m\,y_w)/\sqrt n$,
$b\wedge w=(\sqrt m\,g_w+y_w)/\sqrt n$ transforms it into the stated block,
with $S=A_0+mp/3$, $\eta=\sqrt m(p/3-A_0/2)$ and $\tau=p/3$. Exchangeability
makes these forms proportional to the inner product on $W$, so orthogonal
basis vectors give no cross terms. Writing $\phi_{\rm st}$ and $\phi_{\rm int}$
for the star and internal parts of $\phi_x$,
$\E[\phi_{\rm st}\phi_{\rm int}^{\mathsf T}\mid x_1=\ell]=-(A_\ell/2)D_m$. The
internal cycle space $\Lambda^2W$ has zero incidence, covariance $I/3$, and
no coupling to the other blocks, which gives the last block.
\end{proof}

\subsection{Upper bound}
Let $Q=p+p\upsilon\zeta/m$, which is the right side of~\eqref{eq:one}, and define
\[
\xi=T+\frac{p\upsilon}{\zeta},\qquad \ell_C=\tau+\frac{\eta^2}{\xi-S},\qquad
\mathcal L=QP_0+\xi P_V+\ell_CP_C .
\]
Then $\xi-T=p\upsilon/\zeta>0$ and $\xi-S=p\upsilon(3k+2h\zeta)/(3k\zeta)>0$. Also
$\theta/(\xi-T)=\rho$, $\eta/(\xi-S)=\gamma$ and $Q-p=(\xi-T)\rho^2$, so the blocks of
$\mathcal L-M_1$ factor as
\[
\begin{pmatrix}Q-p&\theta\\\theta&\xi-T\end{pmatrix}=(\xi-T)\binom{\rho}{1}\begin{pmatrix}\rho&1\end{pmatrix},
\qquad
\begin{pmatrix}\xi-S&-\eta\\-\eta&\ell_C-\tau\end{pmatrix}=(\xi-S)\binom{1}{-\gamma}\begin{pmatrix}1&-\gamma\end{pmatrix},
\]
and the last block is $(\xi-S)\gamma^2I\succeq0$. Hence $\mathcal L\succeq M_1$,
and covariance under relabeling gives $\mathcal L\succeq M_i$ for every $i$.
Each edge has $\|P_V(e_i\wedge e_j)\|^2=2/n$ and $\|P_C(e_i\wedge e_j)\|^2=k/n$,
so the edge diagonal of $\mathcal L$ is $(2\xi+k\ell_C)/n$, and
\[
nQ-2\xi-k\ell_C=\frac{p\upsilon}{m\zeta(3k+2h\zeta)}F_{n,r}(\zeta)=0 .
\]
Every diagonal entry of $\mathcal L$ in the feature basis therefore equals $Q$.

To apply this matrix bound to every admissible algorithm, we express its
success in the feature coordinates. This is the one-query case of the
semidefinite characterization of quantum query
algorithms~\cite{BarnumSaksSzegedy2003}. Fix the private randomness of an
algorithm, purify it, and absorb the processing after the call into the
final measurement. Measurements before
the call are input independent and can be deferred. Write the state before
the call as $\sum_{i,j}\ket{i,j}\bigl(\ket+\ket{u_{ij}}+\ket-\ket{v_{ij}}\bigr)+\ket\chi$,
where $\ket\chi$ is the component of branches that skip the call, flagged in
the workspace. In the eigenbasis of the response bit the call fixes the $+$
sector, multiplies the $-$ sector at $(i,j)$ with $i<j$ by $\sgn(x_i-x_j)$,
multiplies the reversed pair by the opposite sign, and fixes self-addresses.
The state after the call is therefore $\ket{\Psi_x}=\sum_\mu(\phi_x)_\mu\ket{a_\mu}$,
where $(\phi_x)_\mu$ are the coordinates of $\phi_x$ and
\[
\ket{a_0}=\ket\chi+\sum_{i,j}\ket{i,j,+}\ket{u_{ij}}+\sum_i\ket{i,i,-}\ket{v_{ii}},
\qquad
\ket{a_{ij}}=\ket{i,j,-}\ket{v_{ij}}-\ket{j,i,-}\ket{v_{ji}}\quad(i<j).
\]
These vectors have disjoint supports on the flag, address and response
registers, so they are orthogonal, and $\sum_\mu\|a_\mu\|^2=1$. For the output measurement
$\{\widetilde E_i\}$ put $(B_i)_{\mu\mu'}=\langle a_\mu|\widetilde E_i|a_{\mu'}\rangle$ and
$\lambda_\mu=\|a_\mu\|^2$. Then $B_i\succeq0$, $\sum_iB_i=\diag(\lambda)$ and
$\sum_\mu\lambda_\mu=1$. The matrices $B_i$ may be complex Hermitian, and since
$M_i$ is real only their real parts contribute to the success. The average success is
\[
\sum_i\tr(M_iB_i)\le\sum_i\tr(\mathcal LB_i)=\tr\bigl(\mathcal L\diag(\lambda)\bigr)=Q .
\]
The worst-case success is at most the average, so $Q_1(n,r)\le Q$.

\subsection{Attainment}
The attaining construction uses the kernels of the slack blocks in the
upper-bound proof. We give the initial state and measurement, then verify
that the measurement is complete and that the success probability equals
the bound for every order.

Let $N=\binom n2=nm/2$, $\alpha=2mR/(n\zeta^2+2mR)$ and $\beta=(1-\alpha)/N$.
Prepare
\[
\ket{\psi_{\rm in}}=\sqrt\alpha\,\ket{1,1,+}+\sqrt\beta\sum_{i<j}\ket{i,j,-}.
\]
One call gives $\ket{\psi_x}=C_\alpha\phi_x$ with
$C_\alpha=\sqrt\alpha P_0+\sqrt\beta(P_V+P_C)$, in the encoding
$\ket0=\ket{1,1,+}$ and $\ket{ij}=\ket{i,j,-}$. In the
decomposition~\eqref{eq:local} define the seed effect
\[
E_1=\frac1n\begin{pmatrix}1&-\sqrt{mR}\\-\sqrt{mR}&mR\end{pmatrix}\oplus
\bigoplus_{a=1}^k\frac{m}{2n}\begin{pmatrix}\gamma^2&\gamma\\\gamma&1\end{pmatrix}\oplus0,
\]
and $E_i=U_\sigma E_1U_\sigma^{\mathsf T}$ for any permutation $\sigma$ with
$\sigma(1)=i$, where $U_\sigma$ fixes $\ket0$ and relabels the oriented edges.
The seed is invariant under the stabilizer of label $1$, so $E_i$ is well
defined. Its traces on the global sectors are $1/n$, $mR/n+km\gamma^2/(2n)=m/n$
and $mk/(2n)$. A label-invariant symmetric matrix on the edge space is a
combination of $I$ and $D_n^{\mathsf T}D_n$, which act as scalars on
$\mathcal V$ and $\mathcal C$, and there is no invariant coupling between
$\ket0$ and the edges. Since $\sum_iE_i$ is invariant and its sector traces
equal the sector dimensions, $\sum_iE_i=I$. On the full query space the
effects are completed by adding $n^{-1}$ times the projector onto the
complement of the encoded subspace, which the prepared states never reach.

Finally $(\alpha/\beta)\rho^2=mR$. The range of the first block of
$C_\alpha E_1C_\alpha$ is the kernel of the first slack block, spanned by
$(1,-\rho)^{\mathsf T}$. The range of each repeated block is the kernel of
the second slack block, spanned by $(\gamma,1)^{\mathsf T}$. So
$\tr\bigl[(\mathcal L-M_1)C_\alpha E_1C_\alpha\bigr]=0$, and by covariance the same
holds for every $i$. The average success is
\[
\sum_i\tr(M_iC_\alpha E_iC_\alpha)=\tr(\mathcal LC_\alpha^2)=\alpha Q+N\beta Q=Q,
\]
and covariance of the state family makes the success equal to $Q$ on every
order.

The closed forms in Corollary~\ref{cor:one} follow from the same cubic.
For $r=1$, $F_{n,1}(z)=2k(z+1)\bigl[nz^2+(3-n)z-3m\bigr]$, whose positive
root $\bigl(n-3+\sqrt{13n^2-18n+9}\bigr)/(2n)$ gives the first closed form.
For $r=n/2$, $h=0$ gives the second closed form.

\end{document}